\documentclass{llncs}

\usepackage{amssymb}
\usepackage{textcomp}
\usepackage{comment}
\usepackage{color}
\usepackage{graphicx}
\usepackage{amsmath}
\usepackage{algorithm}
\usepackage{algpseudocode}

\newcommand{\Endproof}{\hfill$\Box$}

\DeclareMathOperator{\exv}{\mathbb{E}\textrm{ }}

\begin{document}

\title{Quantum  Request-Answer Game with Buffer Model for Online Algorithms}
\author{Kamil~Khadiev}

\institute{
Kazan Federal University, Kazan, Russia
                \\ \email{kamilhadi@gmail.com }}

\maketitle

\begin{abstract}We consider online algorithms as a request-answer game. An adversary that generates input requests, and an online algorithm answers. We consider a generalized version of the game that has a buffer of limited size. The adversary loads data to the buffer, and the algorithm has random access to elements of the buffer. We consider quantum and classical (deterministic or randomized) algorithms for the model.

In the paper, we provide a specific problem (The Most Frequent Keyword Problem) and a quantum algorithm that works better than any classical (deterministic or randomized) algorithm in terms of competitive ratio. At the same time, for the problem, classical online algorithms in the standard model are equivalent to the classical algorithms in the request-answer game with buffer model.
\end{abstract}

\section{Introduction}

One of the applications for online algorithms is optimization problems \cite{k2016}.
The peculiarity is the following. An algorithm reads an input piece by piece
and returns an answer piece by piece immediately, even if an answer can depend on future pieces of the input. The algorithm should return an answer for minimizing an objective function (the cost of an output). The most standard method to define the effectiveness is the competitive ratio \cite{st85,kmrs86}. 

One of the possible point of view to online algorithms is a request-answer game \cite{a1996}. Here we consider a game of an online algorithm and Adversary that holds input. Adversary requests and the algorithm returns answers. We suggest a reversed version of the game. The algorithm asks an input variable and Adversary returns an answer, but as a price for the answer, Adversary asks to return an output variable. The new version of the game is equivalent to the original one, but we can generalize it.  We provide the new model for online algorithms that is called ``Request-answer Game with Buffer''. The model is a game of three players that are an online algorithm, Adversary and Buffer of limited size. The algorithm can do a request of one of two types:
\begin{itemize}
    \item asking Adversary to load the next block of input variables to the Buffer;
    \item request Buffer for one of the holding variables.
\end{itemize}
For some integer parameter $R$, after each $R$ requests Adversary asks an output variable. If the size of Buffer is $1$ and $R=1$, then the model is equivalent to the original one.

{\em Motivation.} Online algorithms have different applications. One of them is making a decision in current time with no knowledge about future data. Another one is processing a data stream and output a result data stream in online fashion, for example, streaming video on web sites and others. Many programming languages like Java, C++ \cite{javaBR,cpp98} and others use buffered data streams that store data in a fast buffer first, and then an algorithm reads data from the buffer. So, our model is like usage of buffered data streams. Additionally, we have asynchronous processing with online output. In other words, we focus on online behavior of the output stream, but when an algorithm reads an input stream, it can skip some data.    

{\em Quantum model.}
In the paper, we consider a quantum version of ``Request-answer Game with Buffer'' model. Quantum computing itself  \cite{nc2010,a2017,aazksw2019part1} is one of the hot topics in computer science. There are many problems where quantum algorithms outperform the best known classical algorithms  \cite{dw2001,quantumzoo,ks2019,kks2019,kms2019}. Superior of quantum  over classical was shown for different computational models like query model, streaming processing models, communication models and others \cite{an2009,av2009,agky14,agky16,aakk2018,aakv2018,kkm2018,kk2017,ikpy2018,l2009,ki2019}.

 Different versions of online quantum algorithms were considered in \cite{kkm2018,aakv2018} including quantum streaming algorithms as online algorithms \cite{kkkrym2017,kk2019disj}, quantum online algorithms with restricted memory \cite{kk2019,kk2020}, quantum online algorithms with repeated test \cite{y2009}.
In these papers, authors show examples of problems that have quantum online algorithms with better competitive ratio comparing to classical online algorithms.

{\em Our results.}
Here we provide a specific problem and a quantum online algorithm in ``Request-answer Game with Buffer'' model for it. We show that the quantum online algorithm has better competitive ratio than any classical (deterministic or randomized) counterpart. The problem is  ``The Most Frequent Keyword Problem''. Questions are strings of length $k$; the problem is searching the most frequent keyword among words of a text and returning it after each word of the text immediately. The problem \cite{ch2008} is one of the most well-studied ones in the area of data streams \cite{m2005,a2007datastreams,bcg2011}. Many
applications in packet routing, telecommunication logging, and tracking keyword queries in search machines are critically
based upon such routines. The similar problem in online fashion was considered in \cite{blm2015}.

  The paper is organized in the following way. Definitions are in Section \ref{sec:prlmrs}. A description of the most frequent question problem and the quantum algorithm for the problem are described in Section \ref{sec:quantum}. Section \ref{sec:classical} contains lower bounds for classical algorithms.
\section{Preliminaries}\label{sec:prlmrs}
%
{\bf An online minimization problem} consists of a set $\cal{I}$ of inputs and a cost function. Each input $I = (x_1, \dots , x_n)$ is a sequence of requests, where $n$ is a length of the input $|I|=n$. Furthermore, a set of feasible outputs (or solutions) ${\cal O}(I)$ is associated with each $I$; an output is a sequence of answers $O = (y_1, \dots, y_n)$. The cost function assigns a positive real value $cost(I, O)$ to $I\in{ \cal I}$ and $O\in{\cal O}(I)$. An optimal  solution for $I\in{\cal I}$ is $O_{opt}(I)=argmin_{O\in{\cal O}(I)}cost(I,O)$.

Let us define an online algorithm for this problem.
%
{\bf A deterministic online algorithm}  $A$ computes the output sequence $A(I) = (y_1,\dots , y_n)$ such that $y_i$ is computed by $x_1, \dots , x_i$.  
%
%
%
We say that $A$ is $c$-{\em competitive} if there exists a constant $\alpha\geq 0$ such that, for every $n$ and for any input $I$ of size $n$, we have: $cost(I,A(I)) \leq c \cdot cost(I,O_{Opt}(I)) + \alpha,$ where $c$ is the minimal number that satisfies the inequality. Also we call $c$ the {\bf competitive ratio} of $A$. If $\alpha = 0, c=1$, then $A$ is optimal.

{\bf A randomized online algorithm} $R$ computes an output sequence
$R^{\psi}(I) = (y_1,\ldots, y_n)$ such that $ y_i$ is computed from $\psi, x_1, \ldots, x_i$, where $\psi$ is the content of the random tape, i. e., an infinite binary sequence, where every bit is chosen uniformly at random and independently of all the others. By $cost(I,R^{\psi}(I))$ we denote the random variable expressing the cost of the solution computed by $R$ on $I$.
$R$ is $c$-competitive in expectation if there exists a constant $\alpha>0$ such that, for every $I$, $\exv[cost(I,R^{\psi}(I))] \leq c \cdot cost(I,O_{Opt}(I)) + \alpha$. We can say that $c$ is expected competitive ratio for the algorithm.

\subsection{Request-answer Game with Buffer Model}
The standard model for online algorithms can be considered as a request-answer game \cite{a1996}. Adversary holds an input, it sends request $x_i$ to an algorithm, and the algorithm sends answer $y_i$. Here Adversary is an ``active'' player that rules the game and the algorithm is a ``passive'' player that answers on each response.

Let us change the point of view to this game. Both are ``active'' players in some sense.

\begin{itemize}
    \item[] {\bf Round $1$}. The algorithm asks an input variable $x_1$. (The algorithm is active on this round).
    \item[]  {\bf Round $2$}.  Adversary asks an output variable $y_1$. (Adversary is active on this round).
    \item[] ...
    \item[] {\bf Round $2i-1$}. The algorithm asks an input variable $x_i$. (The algorithm is active on this round).
    \item[]  {\bf Round $2i$}. Adversary asks an output variable $y_i$. (Adversary is active on this round).
\end{itemize}

It is easy to see that the new game is equivalent to the original game and the standard model.

Let us consider the modification of the game that has a buffer.
Assume that we have a buffer between the algorithm and Adversary. Let a positive integer $K$ be a size of the buffer. Additionally, there is an integer parameter $R\leq K$. 
The algorithm will ask to load data to the buffer by blocks of $K$ variables. Let $i$ be a number of the loading block. 
The algorithm can do the following actions if it is active on some round:

\begin{itemize}
    \item The algorithm asks to erase the buffer and load the next $K$ input variables $x_{i\cdot K+1},\dots,x_{i\cdot K+K}$ to the buffer.  After that, $i$ is increased by $1$. ($i\gets i+1$)
    \item The algorithm requests any variable from the buffer. We consider a query model (decision tree model) for the algorithm that queries variables from the buffer.
\end{itemize}

The game has the following scenario:

\begin{itemize}
  \item[] {\bf Round $0$}. We initialize $i\gets 0$
    \item[] {\bf Round $1$}. The algorithm is active and it does the possible actions that were described before.
    \item[] {\bf Round $2$}. The algorithm is active and it does the possible actions that were described before.
     \item[] ...
    \item[] {\bf Round $R$}. The algorithm is active and it does the possible actions that were described before.
    \item[] {\bf Round $R+1$}.  Adversary is active. He asks output variables $y_{1},\dots,y_{R}$.
    \item[] ...
    \item[] {\bf Round $(R+1)\cdot j +1$}. The algorithm is active and it does the possible actions that were described before.
    \item[] {\bf Round $(R+1)\cdot j +2$}. The algorithm is active and it does the possible actions that were described before.
     \item[] ...
    \item[] {\bf Round $(R+1)\cdot j +R$}. The algorithm is active and it does the possible actions that were described before.
    \item[] {\bf Round $(R+1)\cdot j +R+1$}. Adversary is active. He asks output variables $y_{j\cdot R+1},\dots,y_{j\cdot R+R}$.
    
\end{itemize}

\noindent
{\bf Comment.} {\em In the case of $K=1$ and $R=1$, the new model is equivalent to the standard online algorithms model.}

In the randomized case, an algorithm that requests data from the buffer can be randomized, and we use a randomized query model in that case. We consider an expected competitive ratio for the model as for the standard model of randomized online algorithms. At the same time, the loading the next block to the buffer is deterministic action.

In the quantum case, an algorithm that requests data from the buffer can be quantum, and we use a quantum query model in that case. Because of the probabilistic behavior of quantum algorithms, we also consider an expected competitive ratio for the model.  At the same time, the loading the next block to the buffer is deterministic action.

We skip details of the quantum model and quantum algorithms here because we use them as quantum subroutines and the rest part is classical. More details on quantum query model and quantum algorithms can be found in \cite{nc2010,a2017,aazksw2019part1} 

\section{A Quantum Algorithm for The Most Frequent Keyword Problem}\label{sec:quantum}

Let us present the problem formally.

    {\bf Problem} For some positive integers $m,d$ and $k$, the input is 
    \[I=(s^1,\dots,s^d,x^1,\dots, x^m).\]
    Here $(s^1,\dots,s^d)$ is a sequence of strings that are interesting keywords for us in the input, $s^j=(s^j_1,\dots,s^j_k)\in \{0,1\}^k$, for $j\in\{1,\dots,d\}$. Strings $x^1,\dots, x^m$ are words of a text, $x^j=(x^j_1,\dots,x^j_k)\in \{0,1\}^k$, for $j\in\{1,\dots,m\}$. The input length is $n=(m+d)\cdot k$. A frequency of a string $t\in\{0,1\}^k$ is $f(t)=\frac{\#(t)}{m}$, where $\#(t)=|\{i:t=x^i, i\in\{1,\dots,m\}\}|$ is a number of occurrence of $t$ in $(x^1,\dots, x^m)$.
     The index $i_0$ of the most frequent string $s^{i_0}$ is such that $f(s^{i_0})=\max\limits_{i\in\{1,\dots,d\}}f(s^i)$ and $i_0$ is minimal.
    We should return index $i_0$ after reading each string $x^j$.  So, the right answer that returns offline algorithm is $(z_1,\dots,z_n)$ where $z_{(j+d)\cdot k}=i_0$ for $j\in\{1,\dots,m\}$ and other output variables are not considered.
    
    The cost of an output $O=(y_1,\dots,y_n)$ is \[cost(I,O)=1 + m-\sum_{j=1}^{m}\delta(y_{(j+d)\cdot l},i_0)\]
    Here $\delta(a,b)=1$ if $a=b$ and $\delta(a,b)=0$ if $a\neq b$
    
\subsection{Quantum Algorithm}

Firstly, we discuss a quantum subroutine that compares two strings of length $l$ for some integer $l>0$.

\subsubsection{The Quantum Algorithm for Two Strings Comparing}\label{sec:compare}
Assume that the subroutine is $\textsc{Compare\_strings}(s,t)$ and it compares $s$ and $t$ in lexicographical order. It returns:
\begin{itemize}
    \item $-1$ if $s<t$;
    \item $0$ if $s=t$;
    \item $1$ if $s>t$.
\end{itemize}

As a base for our algorithm, we will use the algorithm of finding the minimal argument with $1$-result of a Boolean-value function. Formally, we have:
\begin{lemma}\cite{kkmyy2020}\label{lm:first-one}
Suppose, we have a function $f:\{1,\dots,N\}\to \{0,1\}$ for some integer $N$. There is a quantum algorithm for finding $j_0=\min\{j\in\{1,\dots,N\}:f(j)=1\}$. The algorithm finds $j_0$ with query complexity $\sqrt{N}$ and error probability that is at most $\frac{1}{2}$.
\end{lemma}

Let us choose the function $f(j)=(s_j\neq t_j)$. So, we search $j_0$ that is the index of the first unequal symbol of the strings. We search $j_0$ among indexes $1,\dots \min(|s|,|t|)$, where $|s|$ is a length of $s$. Then, we can claim that $s$ precedes $t$ in lexicographical order iff $s_{j_0}$ precedes $t_{j_0}$ in the alphabet for strings. If there are no unequal symbols, then we have one of three options:
\begin{itemize}
    \item if $|s|<|t|$, then $s<t$;
    \item if $|s|>|t|$, then $s>t$;
    \item if $|s|=|t|$, then $s=t$.
\end{itemize}    We use $\textsc{The\_first\_one\_search}(f,N)$ as a subroutine from Lemma \ref{lm:first-one}, where $f(j)=(s_j\neq t_j)$. Assume that this subroutine returns $N+1$ if it does not find any solution.

We apply the standard technique of boosting success probability that was used, for example, in \cite{ks2019}. So, we repeat the algorithm $3\log_2 m$ times and return the minimal answer, where $m$ is a number of strings in the sequence $(x^1,\dots x^m)$. In that case, the error probability is $O\left(\frac{1}{2^{3\log m}}\right)=O\left(\frac{1}{m^3}\right)$.

Let us present the algorithm.
\begin{algorithm}
\caption{$\textsc{Compare\_strings}(s,t,k)$. The Quantum Algorithm for Two Strings Comparing.}\label{alg:strcmp}
\begin{algorithmic}
\State $N\gets min(|s|,|t|)$
\State $j_0 \textsc{The\_first\_one\_search}(f,N)$\Comment{The initial value}
\For{$i \in \{1,\dots,3\log_2 m\}$}
\State $j\gets \textsc{The\_first\_one\_search}(f,N)$
\If{$j\leq k$ and $s_j \neq s_t$}
\State $j_0 \gets  \min(j_0, j)$
\EndIf
\EndFor

\If{$j_0=N+1$ and $|s|=|t|$} 
\State $result \gets 0$\Comment{The strings are equal.}
\EndIf
\If{$((j_0\neq N+1 )$ and $(s_{j_0}<t_{j_0}))$ or $((j_0= N+1 )$ and $(|s|<|t|))$}
\State $result \gets -1$ \Comment{$s$ precedes $t$.}
\EndIf
\If{$((j_0\neq N+1 )$ and $(s_{j_0}>t_{j_0}))$ or $((j_0= N+1 )$ and $(|s|>|t|))$}
\State $result \gets 1$ \Comment{$t$ succeeds $s$.}
\EndIf
\State \Return $result$
\end{algorithmic}
\end{algorithm}

Let us discuss the property of the algorithm:
\begin{lemma}\label{lm:strcmp}
Algorithm \ref{alg:strcmp} compares two strings $s$ and $t$ in lexicographical order with  query complexity $O(\sqrt{\min(|s|,|t|)}\log m)$ and error probability $O\left(\frac{1}{m^3}\right)$.
\end{lemma}
\begin{proof}
The correctness of the algorithm follows from description and lexicographical order.

Let us discuss the error probability.
The algorithm has error iff there are error in all $3\log_2 m$ invocations of $\textsc{The\_first\_one\_search}$ algorithm. The probability of such event is at most $0.5^{3\log_2 m}=O\left(\frac{1}{m^3}\right)$.
\Endproof
\end{proof}

\subsubsection{A Quantum Algorithm in Request-answer Game with Buffer Model}

Firstly, we present an idea of the algorithm.

We use the well-known data structure a self-balancing binary search tree. As an implementation of the data structure, we can use the AVL tree \cite{avl62,cormen2001} or the Red-Black tree \cite{g78,cormen2001}. Both data structures allow us to find and add elements in $O(\log N)$ running time, where $N$ is a size of the tree.

The idea of the algorithm is the following. We store a triple $(i,s,c)$ in a vertex of the tree, where $i$ is the minimal index of a string from $\{s^1,\dots,s^d\}$ such that $s=s^i$ and $c$ is a number of occurrences of the string $s$ among $\{x^1,\dots,x^m\}$. We assume that a triple $(i,s,c)$ is less than a pair $(i', s',c')$ iff $s$ precedes $s'$ in the lexicographical order. So, we use $\textsc{Compare\_strings}(s,s',k)$ subroutine as the comparator of the vertexes.  The tree represents a set of unique strings from $\{s^1,\dots,s^d\}$ with a number of occurrences among $(x^1,\dots,x^m)$.

Firstly, we load all strings $s^1,\dots,s^d$ one by one to Buffer and add a vertex $v=(j, s^j,0)$ for each string $s^j$ to the tree, here $j\in\{1,\dots,d\}$. We add only one node for each duplicate strings from $s^1,\dots,s^d$ if they exist. The index $j$ in $v$ stores the index of $s^j$ and if there is no a vertex that corresponds to $s^j$, then $j$ is a minimal index from all possible indexes. $0$ in $v$ means that initially we assume that $s^j$ does not occurs among $(x^1,\dots,x^m)$. 

Secondly, we load questions (strings) from $x^1$ to $x^m$ one by one to Buffer and search them in our tree. We increase the number of occurrences. If the string was not found in the tree, then it is not a keyword, i.e. it does not belong to $s^1,\dots s^d$ and we skip it. At the same time, we store \[(i_{max},s,c_{max})=argmax_{(i,t,c)\mbox{ {\it in the tree} }}c\] and recalculate it in each step. When Adversary requests an output variable, then we return $i_{max}$.

Let us present the algorithm formally. Let $BST$ be  a self-balancing binary search tree such that:
\begin{itemize}
    \item $\textsc{Find}(BST, x^i)$ finds a vertex $(j,s,c)$ such that $s=x^i$, or $NULL$ if $x^i$ was not found. The standard algorithm for searching $x^i$ in the tree is comparing with elements of vertexes and moving by the tree according to the result of the comparison. When we invoke the  $\textsc{Compare\_strings}$ subroutine, we request a variable from Buffer for checking a symbol of $x^{i}$ and request to memory when we check a symbol of a string that is stored in a vertex.
    \item $\textsc{Add}(BST, j, s^j)$ adds a vertex $(j,s^j,0)$ to the tree if a  vertex with $s^j$ does not exist; and does nothing otherwise.
     \item $\textsc{Init}(BST)$ initializes an empty tree.
\end{itemize}

\begin{algorithm}[h]
\caption{A Quantum Algorithm for The Most Frequent Keyword Problem.}\label{alg:qmain}
\begin{algorithmic}
\State $\textsc{Init}(BST)$\Comment{The initialization of the tree.}
\State $c_{max}\gets 1$\Comment{The maximal number of occurrences.}
\State $i_{max}\gets 1$\Comment{The index of most frequent question.}
\State $step \gets 0$
\For{$j \in \{1,\dots,d\}$}
\State $\textsc{Load\_To\_Buffer}$\Comment{Load $s^j$ to Buffer}
\State $t\gets``''$\Comment{Initially $t$ is an empty string}
\For{$q \in \{1,\dots,k\}$}\Comment{Reading the string $t$}
\State $t\gets t +\textsc{Request}(q)$\Comment{Requesting $q$-th variable from Buffer and appending the variable to $t$}
\EndFor
\State $\textsc{Add}(BST, j, t)$\Comment{Adding the string $t=s^j$ to the tree as a vertex $(NULL, t, 0)$}
\EndFor
\For{$j \in \{1,\dots,m\}$}
\State $\textsc{Load\_To\_Buffer}$\Comment{Load $x^i$ to Buffer}
\State $v=(i,t,c) \gets \textsc{Find}(BST, x^j)$\Comment{Searching $x^i$ in the tree.}
\If{$v\neq NULL$}\Comment{If $x^i$ belongs to $(s^1,\dots,s^d)$}
\State $c\gets c+1$\Comment{Updating the vertex by increasing the number of occurrences.}
\State $v\gets (i,t,c)$\Comment{Updating the vertex by the new values} 
\If{$c>c_{max}$}\Comment{Updating the maximal value.}
\State $c_{max} \gets c$
\State $i_{max} \gets i$
\EndIf
\EndIf
\EndFor
\If{Adversary request an output variable}
\Return $i_{max}$
\EndIf

\end{algorithmic}
\end{algorithm}

Let us discuss the property of the algorithm.
\begin{theorem}\label{th:qfreq-compl}
The expected competitive ratio $c$ for Algorithm \ref{alg:qmain} is at most ${\cal C}_Q$ where 
\[{\cal C}_Q=O\left(1+\frac{(m\log m)\cdot (\log d)}{\sqrt{k}}\right).\]
\end{theorem}
\begin{proof}
The correctness of the algorithm follows from the description.
Let us discuss the query complexity of $\textsc{Find}(BST, x^j)$. The procedure requires $O(\log d)$ comparing operations $\textsc{Compare\_strings}(x^{j},s^{i'},k)$. Due to Lemma \ref{lm:strcmp}, each comparing operation requires $O(\sqrt{k}\log m)$ queries. The total query complexity of the $\textsc{Find}$ procedure is  $O\left(\sqrt{k}(\log m)\cdot (\log d)\right)$. So, the algorithm  checks all $x^1,\dots, x^m$ in$O\left(m\sqrt{k}(\log m)\cdot (\log d)\right)$ rounds and after that returns right answers for the requests of Adversary. Therefore, the first $O\left(\frac{m\sqrt{k}(\log m)\cdot (\log d)}{k}\right)=O\left(\frac{m(\log m)\cdot (\log d)}{\sqrt{k}}\right)$ ``significant'' output variables can be wrong and others are right. We call output variable $y_{(j+d)\cdot k}$, for $j\in\{1,\dots,m\}$, as ``significant'' because the cost depends on these variables. Hence, the cost is at most $1+O\left(\frac{m(\log m)\cdot (\log d)}{\sqrt{k}}\right)$.

Let us discuss the error probability. Events of error in the algorithm are independent. So, all events should be correct. Due to  Lemma \ref{lm:strcmp}, the probability of correctness of one event is $1-\left(1-\frac{1}{m^3}\right)$. Hence, the probability of correctness of all $O(m\log m)$ events is at least $ 1-\left(1-\frac{1}{m^3}\right)^{\gamma\cdot m\log m}$ for some constant $\gamma$.

Note that 
\[ \lim\limits_{n\to \infty} \frac{\left(1-\frac{1}{m^3}\right)^{\gamma\cdot m\log m}}{1/m}<1;\]
 Hence, the total error probability is at most $O\left(\frac{1}{m}\right)$.

In a case of an error, all ``significant'' output variables can be wrong.

Therefore, the expected competitive ratio of the algorithm is at most

\[{\cal C}_Q=\frac{O(\frac{m-1}{m})\cdot\left(1+O\left(\frac{m(\log m)\cdot (\log d)}{\sqrt{k}}\right)\right)+ O\left(m\cdot\frac{1}{m}\right)}{1}=O\left(1+\frac{m(\log m)\cdot (\log d)}{\sqrt{k}}\right).\]

\Endproof
\end{proof}

\section{Lower Bounds for Classical Algorithms for The Most Frequent Keyword Problem}\label{sec:classical}
There is an input $I_B$ such that any classical (deterministic or randomized) algorithm returns output with the cost at least $O(m)$.
\begin{theorem}
 Any randomized algorithm for the problem has competitive ratio $c$ at least ${\cal C}_R=O(m)>{\cal C}_Q$ in a case of $(\log_2 m)\cdot (\log_2 d)=o(\sqrt{k})$.
\end{theorem}
\begin{proof}
Let us show that the problem is equivalent to unstructured search problem.
Assume that $m=2t$ for some integer $t$. 
Then, let $x^{t+1},\dots,x^{2t}=0^k$ where $0^k$ is a string of $k$ zeros.
We have two cases for other string:
\begin{itemize}
    \item {\bf case 1:} $x^{1},\dots,x^{t}=1^k$;
    \item {\bf case 2:} there are $z\in\{1,\dots,t\}$ and $u\in\{1,\dots,k\}$ such that $x^{z}_{u}=0$ and $x^{z}_{u'}=1$ for all $u'\in\{1,\dots,u-1,u+1,\dots,k\}$, $x^{z'}=1^k$ for $z'\in\{1,\dots,t\}\backslash \{z\}$.
\end{itemize}
Let $d=2$, $s^1=0^k$ and $s^2=1^k$.

In the first case, the answer is $1^k$. In the second case, the answer is $0^k$. 
Therefore, the problem is equivalent to search $0$ among the first $tk=mk/2$ variables.

Due to \cite{bbbv1997}, the randomized query  complexity of unstructed search among $mk/2$ is $\Omega(mk)$.

In a case of odd $m$, we assign $x^m=1^{k/2}0^{k/2}$, and it is not used in the search. Then, we can consider only $m-1$ strings. So, $m-1$ is even.

Suppose, we have a randomized algorithm $A$ for finding the most frequent question that uses $o(mk)$ queries to buffer when it reads $x^1,\dots,x^m$. Then, Adversary can construct the input $I_B$ such that $A$ obtains a wrong answer.

Therefore, all ``significant'' output variables will be wrong and $cost(I_B,A(I_B))=1+m$. The competitive ratio in that case is ${\cal C}_R=m+1$.

If the algorithm do $O(mk)$ queries to Buffer for computing answer, then $O(m)$ ``significant'' output variables should be returned before getting a right answer. Therefore, $cost(I_B,A(I_B))=O(m)$ and ${\cal C}_R=O(m)$.

In the case of $(\log_2 m)\cdot (\log_2 d)=o(\sqrt{k})$ we have
\[{\cal C}_Q=O\left(1+\frac{m(\log_2 m)\cdot (\log_2 d)}{\sqrt{k}}\right)=o(m)<O(m)={\cal C}_R.\]
\Endproof
 \end{proof}
\section{Conclusion}
We consider a new setting or new model for online algorithms that is useful for real world problems. 
We show that in the case of $(\log_2 m)\cdot (\log_2 d)=o(\sqrt{k})$ the quantum algorithm shows a better competitive ratio than any classical (deterministic or randomized) algorithm. Note that this setting is reasonable. 

\paragraph*{Acknowledgements}
 The research was funded by the subsidy allocated to Kazan Federal University for the state assignment in the sphere of scientific activities, project No. 0671-2020-0065.

We thank Farid Ablayev and Aliya Khadieva from Kazan Federal University for helpful discussions.
\bibliographystyle{plain}
\bibliography{tcs}

\end{document}